\documentclass[pdflatex,sn-mathphys,Numbered]{sn-jnl}

\usepackage{amsmath,amssymb,amsfonts}%
\usepackage{amsthm}%
\usepackage{xcolor}%
\usepackage{textcomp}%
\usepackage{manyfoot}%
\usepackage{algorithm}%
\usepackage{algorithmicx}%
\usepackage{algpseudocode}%

\theoremstyle{thmstyleone}%
\newtheorem{theorem}{Theorem}%

\theoremstyle{thmstyletwo}%
\newtheorem{remark}{Remark}%

\theoremstyle{thmstylethree}%
\newtheorem{definition}{Definition}%

\newcommand{\R}{\mathbb{R}}

\newcommand{\Z}{\mathbb{Z}}
\newcommand{\N}{\mathbb{N}}

\newcommand{\D}{\mathbb{D}} 
\renewcommand{\P}{\mathbb{P}}
\newcommand{\E}{\mathbb{E}}

\newcommand{\ind}{\text{\usefont{U}{bbold}{m}{n}1}}

\newcommand{\Fc}{\mathcal{F}}

\begin{document}

\title[Perfect simulation of load balancing networks]{Perfect simulation of Markovian load balancing queueing networks in equilibrium}
\author*{\fnm{Carl}\sur{Graham}\footnote{e-mail: \href{mailto:carl.graham@polytechnique.edu}{carl.graham@polytechnique.edu} --- July 3, 2024}} 
\affil{CMAP, CNRS, \'Ecole polytechnique, Institut Polytechnique de Paris, 91120 Palaiseau, France, and Inria}

\abstract{We define a wide class of Markovian load balancing networks of identical single-server infinite-buffer queues. These networks may implement classic parallel server or work stealing load balancing policies, and may be asymmetric, for instance due to topological constraints.
The invariant laws are usually not known even up to normalizing constant. 
We provide three perfect simulation algorithms enabling Monte Carlo estimation of quantities of interest in equilibrium. 
The state space is infinite, and the algorithms use a dominating process provided by the network with uniform routing, in a coupling preserving a preorder which is related to the increasing convex order.
It constitutes an order up to permutation of the coordinates, strictly weaker than the product order. The use of a preorder is novel in this context. 
The first algorithm is in direct time and uses Palm theory and acceptance-rejection. Its duration is finite, a.s., but has infinite expectation.
The two other algorithms use dominated coupling from the past; one achieves coalescence by simulating the dominating process into the past until it reaches the empty state, the other, valid for exchangeable policies, is a back-off sandwiching method.  Their durations have some exponential moments.}



\maketitle

\section{Introduction}\label{sec-intro}

This paper provides perfect simulation algorithms for the invariant laws (or stationary distributions) of Markovian queueing networks managed by load balancing policies (or protocols, algorithms, etc.). 
It devises a general abstract definition of these networks, whereas most previous works consider specific examples only. 

These invariant laws are usually not known even up to normalizing constants, which precludes the use of Gibbs samplers.
The algorithms allow to draw samples for practical purposes such as the Monte Carlo estimation of quantities of interest in equilibrium, \emph{e.g.}, quality of service (QoS) indicators for performance evaluation.
They are ``perfect'' up to  real-life constraints on computer precision, running times, pseudo-random number generators, etc.

Each network consists of $c\ge2$ identical single-server infinite-buffer queues, and its state is given by the queue length vector.
Tasks arrive according to a Poisson process of intensity $\lambda>0$.
Each task has an exponential service duration of parameter $1$ (w.l.o.g., up to time-scale) after which it exits the system. 
Queueing discipline is work conserving, thus the total service rate is equal to the number of non-empty queues (busy servers). 
The initial state, arrival stream, and services are independent.
The stability condition $\lambda <c$ will always be assumed to hold.

Load balancing policies observe the network state and strive to allocate tasks to server queues so as to minimize server idleness. 
They may have to satisfy varied constraints, such as partial observation, distributed functioning, and limited overhead. Parallel queue networks allocate each task at arrival to a server queue and leave it there until served~\cite{Gupta-Walton2019,vanderBoor2022}, but other kinds of networks may implement work stealing or jockeying at instants of service completion~\cite{Elsayed-Bastani1985,Zhao-Grassmann1995,Mitzenmacher2001,Tchiboukdjian-etal2010,Suksompong-etal2016,Kielanski-VanHoudt2021}. 

We shall determine a wide class of load balancing networks allowing quite general policies. 
These policies are possibly asymmetric in terms of the queues, for instance they may respect a directed graph network topology, and may thus yield non-exchangeable queue length vector processes.
Most previous results are proved for specific exchangeable networks in suitable asymptotic regimes. 
Many interesting networks are asymmetric, non-asymptotic results are often needed, and simulation studies such as those in~\cite{Turner1998} are then quite useful.
See~\cite[Sect.~6]{vanderBoor2022}, \cite{Budhiraja2019}, \emph{e.g.}

We shall provide three perfect simulation algorithms.  
The first one is a direct-time acceptance-rejection method based on Palm theory. Its run time is finite but has infinite mean, and it cannot be implemented in practice.
The other two use backward couplings and have run times with some exponential moments. 
Backward couplings have long been used for invariant laws~\cite{Loynes1962}. 
Coupling From The Past (CFTP) was devised by Propp and Wilson~\cite{Propp1996} for perfect simulation of the invariant law of a Markov chain on a finite state space with a largest and a least element, and
has been extended for infinite state spaces into Dominated CFTP (DomCFTP) using suitable dominating processes simulated backward and forward in time, see~\cite{Connor-Kendall2015,Kendall1998,Kendall-Moller2000,Thonnes2000,Kendall2005,Huber2016}.

We were inspired by the use in Connor and Kendall~\cite{Connor-Kendall2015} of DomCFTP 
for a $M/G/c$ [FCFS] queue, and we refer to this paper for interesting background on the topic.
The service distribution has a second moment, thus the emptying time from equilibrium has a first moment. 
The arrival rate is $\lambda$, the mean service time is $1/\mu$, the offered traffic is $\lambda/\mu$, 
and the queue is stable if and only if $\lambda/\mu<c$. 
The Kiefer-Wolfowitz workload vectors record the evolution of the $M/G/c$ [FCFS] queue by a Markov chain. 
The dominating process is given by a $M/G/c$ queue with uniform routing,
called random assignment and denoted by $M/G/c$ [RA], which can be viewed as a system of $c$ i.i.d.\ $M/G/1$ queues
with arrival rates $\lambda/c$ denoted by $[M/G/1]^c$. 
The invariant law of the $M/G/1$ queue can be perfectly simulated using the Pollaczek-Khintchine formula.
The workload of the $M/G/1$ queue does not depend on the queueing discipline, and the dynamic reversibility of the 
$M/G/1$ [PS] queue is used to simulate the dominating process backward in time. 
Much care is needed to ensure that the coupling between $M/G/c$ [FCFS] and $[M/G/1]^c$ is preserved 
when simulating forward in time.

The present paper is organized as follows.
Section~\ref{sec-defs} provides examples of load balancing networks and then provides a general definition. 
Section~\ref{sec-coupl-preord} constructs a coupling in which the uniform routing (UR) network has a task backlog not less than an arbitrary load balancing network, in a preorder related to the increasing convex order which constitutes an order up to permutation of the coordinates strictly weaker than the product (coordinate-wise) order. 
To the best of our knowledge, the use of such a preorder is novel for perfect simulation.
Section~\ref{sec-sim-sym} explains how to manage efficiently in code the coupled processes and defines an embedded Markov chain which has same invariant law and is simpler to simulate.
Section~\ref{sec-Palm} exploits the domination by the UR network for an acceptance-rejection method based on Palm theory, inspired by~\cite{Sigman2012}.
Section~\ref{sec-DomCFTP} provides the two DomCFTP methods inspired by~\cite{Connor-Kendall2015}. 
The dominating UR network is reversible in equilibrium and thus simulatable backward  in time, and the preservation of the coupling through the backward and forward simulations is obtained by considering the law of an augmented Markov chain.
The first method achieves coalescence by simulating the dominating process into the past until it reaches the empty state. 
The second is valid only for exchangeable networks and uses a back-off sandwiching method which is likely to be much quicker.

\section{Load balancing networks}
\label{sec-defs}
\subsection{Examples of parallel server load balancing policies}
\label{sec-examples} 

Parallel server policies allocate each task at arrival to a server queue and leave it there until it ends its service
and departs the network. 

Uniform Routing (UR).
At each arrival, the policy chooses a queue uniformly at random and allocates the task to it. 
This blind policy is perhaps the only one for which the invariant law is known:
the $c$ queues receive independent Poisson arrivals of intensities $\lambda/c$ and constitute a system of 
i.i.d.\ $M/M/1$ queues. Under the ergodicity condition $\lambda<c$, the invariant law is $\gamma^{\otimes c}$ where  $\gamma$ is the geometric law given by $\gamma(n) \triangleq   (1-(\lambda/c) ) (\lambda/c) ^n$ for $n\in\N_0$.
This allows performance evaluation in equilibrium, \emph{e.g.}, the probability that a given queue contains at least $n$ tasks and the limit proportion of such queues when $c$ and  $\lambda$ go to infinity with $\lambda/c$ fixed is given by $\gamma([n,\infty))=(\lambda/c)^n$.

Join the Shortest Queue (JSQ).
At each arrival, the task is routed to one of the queues of minimal length in the network.
Here and below, ties may be broken arbitrarily, \emph{e.g.}, uniformly at random.
JSQ requires high overhead in form of intensive signaling and bookkeeping, 
and is practicable only in certain applications such as call centers or parallel computing of moderate sizes.
JSQ is optimal among parallel queue policies using only the past queue length information,
see~\cite{Winston1977} and the coupling proof in~\cite{Liu1995} inspired by previous work of S.~Foss.
The invariant law is unknown, even though complex analysis provides useful transforms
of these laws for $c=2$, see~\cite{Fayolle-Iasnogorodski1979,Kurkova-Suhov2003}.

Power of $d$ Choices (JSQ($d$)). 
At each arrival, $d\ge2$ queues are chosen uniformly at random, with or without replacement according to the variant 
chosen, and the task is routed to one of these of minimal length. 
The degenerate case $d=1$ corresponds to UR and the case $d=c$ without replacement to JSQ. 
JSQ($d$) can be seen as a low overhead proxy for JSQ, and has long been intensively studied.
 
Idle Queue First (IQF).
At each arrival, the task is routed to an empty queue (\emph{i.e.}, idle server)  in the network if any are present, else to a
uniformly chosen queue. This can be seen as a low overhead proxy for JSQ in which queues only need to signal their 
changes of status between empty and non-empty. This policy was investigated in~\cite{Lu2011} 
under the name Join Idle Queue.

Idle One First (IOF).
At each arrival,  the task is routed to an empty queue in the network if any, else to a queue of length one if any, 
else to a uniformly chosen queue. This can be seen as a low overhead proxy for JSQ in which queues only 
need to signal their changes of status between empty, containing one task, and other. 
This policy was introduced and investigated in~\cite{Gupta-Walton2019} under the acronym I1F, see below.  
  
\subsection{Asymptotic performance evaluation}
\label{sec-perf-eval}

Asymptotic performance evaluation yields guidelines and understanding. 
It is usually applied to an exchangeable network in which a reduced low-dimensional, ideally one-dimensional, representation exists yielding a tractable limit process.

The $M/M/c$ queue has an optimal server utilization.
It is called a ``centralized system'' in~\cite{Gupta-Walton2019,vanderBoor2022} and cannot be matched by parallel server queues, but can be by a network implementing IQF at arrivals and work stealing in which any server finishing service of the last task in its queue simultaneously steals unattended work from another queue if any. 
It is interpreted as a ``Join the Shortest Workload'' policy by the Kiefer-Wolfowitz construction \cite[Sect.~2]{Connor-Kendall2015}.
Its invariant law is explicit up to a normalizing constant which is known for all practical purposes from the well-tabulated Erlang $C$ formula. This allows for precise performance evaluation and provides a benchmark.

JSQ($d$) has been investigated in the mean-field regime in which $c$ and $\lambda$ go to infinity with $\lambda/c$ and $d\ge2$ kept fixed. The equilibrium proportion of queues containing at least $n$ tasks converges to $((\lambda/c)^{d^n}-1)/(d-1)$ with hyper-exponential decay in $n$, much faster than the exponential decay $(\lambda/c)^n$ for UR.
See~\cite{Vvedenskaya1996, Mitzenmacher1996,Graham2000,Graham2005}, \emph{e.g.}
In this regime, the asymptotic performance of IQF is perfect since the limit proportion of empty queues is non-zero and thus every task starts service at arrival.

Gupta and Walton~\cite{Gupta-Walton2019} investigate parallel server policies in a much more heavily loaded regime, the nondegenerate slowdown regime~\cite{Atar2012} in which $\lambda$ and $c$ go to infinity with $c - \lambda$ going to $\alpha>0$. They prove that the performances of JSQ and IOF are comparable and are within 15\% of the performance of $M/M/c$, while the performance of IQF is within 100\% of the performance of $M/M/c$.
In this regime, IOF is asymptotically optimal and better than IQF, and it is pointless to implement 
more subtle policies, for instance with higher thresholds.

The wide survey~\cite{vanderBoor2022} studies parallel server networks in the fluid regime as well as in the Halfin-Whitt regime in which $\lambda$ and $c$ go to infinity with $\frac{1}{\sqrt{c}}(c - \lambda)$ going to $\beta>0$. 
In these regimes, IQF is already asymptotically optimal.

\subsection{General definitions for load balancing networks}

A filtered probability space $(\Omega, \Fc, (\Fc_t)_{t\in\R_+},\P)$ satisfying the usual assumptions is given.
It supports all random elements, all random processes are adapted, and
all Poisson process are $(\Fc_t)_{t\in\R_+}$-Poisson processes.
Random elements are assumed to be independent from each other if not stated otherwise.
Let $\N_0 \triangleq \{0,1,\dots\}$ and $(e_i)_{1\le i \le c}$ denote the canonical basis of 
$\R^c$ and $\mathfrak{S}_c$ the permutation group of $\{1,\dots,c\}$.

We devise a workable general definition of load balancing networks, instead of focusing on specific examples.
It allows any classic load balancing policy. 

\begin{definition}[Load balancing network]
\label{def-class}
Tasks arrive according to a $(\Fc_t)_{t\in\R_+}$-Poisson process of intensity $\lambda>0$ on $c\ge2$ single-server 
infinite-buffer queues. Service duration is exponential with parameter $1$.
The state of the network is described by the queue length (vector) process with sample paths in $\D(\R_+,\N_0^c)$ given by 
\[
X \triangleq (X(t))_{t\in\R_+}\,,
\quad
X(t) \triangleq (X^i(t))_{1 \le i \le c}\,,
\]
where $X^i(t)$ is the length of (number of tasks in) queue $i$ at time $t$. 
The load balancing networks under consideration respect the following rules. 
\begin{enumerate}
\item
The process $(X(t))_{t\in\R_+}$ is $(\Fc_t)_{t\in\R_+}$-Markov.
\item\label{task-arr}
At each arrival, the task is routed to a queue not longer than a uniformly chosen queue.
Specifically,  if $t$ is a task arrival instant then the policy chooses $j$ uniformly in $\{1,\dots,c\}$
and determines a queue of index~$i$ such that 
\[
X^i(t-) \le X^j(t-)\,, \qquad X(t) = X(t-) +e_i\,.
\]
\item\label{serv-compl}
The queueing discipline in each queue is work conserving. Thus, the total service intensity is equal to  
the number of non-empty queues and the subsequent service completion concerns a uniformly 
chosen queue among these. 
\item\label{work-steal}
At each instant of service completion, one task may be reallocated from a queue 
which just before had a length not less than the queue completing service to the latter.
Specifically, if $t$ is a service completion instant and queue~$j$ is completing service
then the policy determines a queue of index~$i$ such that 
\[
X^i(t-) \ge X^j(t-)\ge 1\,, \qquad X(t) = X(t-) -e_i\,.
\]
\end{enumerate}

The policy and the network are said to be \emph{exchangeable} if the decision taking is invariant in law under permutations 
of the queue indices. In this case, if $X(0)$ is exchangeable then $X(t)$ is exchangeable, and it is so in equilibrium. 

The uniform routing (UR) network, for which $i=j$ in Rules~\ref{task-arr} and~\ref{work-steal}, is an exchangeable load balancing network with queue length process denoted by 
\[
Y \triangleq (Y(t))_{t\in\R_+}\,,
\quad
Y(t) \triangleq (Y^i(t))_{1 \le i \le c}\,.
\]
\end{definition}

Poisson process splitting implies that $Y \triangleq (Y(t))_{t\in\R_+}$ constitutes a system of $c$ i.i.d.\ $M/M/1$ 
queue processes with arrival rates $\lambda/c$ and service rates~$1$. 
It is positive recurrent if and only if $\lambda <c$ and then it is reversible in equilibrium
and has simulatable invariant law $\gamma^{\otimes c}$, the $c$-fold product of the geometric law
\begin{equation}
\label{geom-law}
\gamma \triangleq (\gamma(n))_{n\in\N_0}\,,
\quad\gamma(n) \triangleq   (1-(\lambda/c) ) (\lambda/c) ^n\,.
\end{equation}

Rule~\ref{task-arr} is satisfied by all policies in Section~\ref{sec-examples}.
Rule~\ref{work-steal} is trivially satisfied by any parallel server policy, but allows work stealing or jockeying in which the queue of a server completing a service may be immediately joined by a task from a longer queue, such  instantaneous jumps not being recorded explicitly by the queue length process. 
The implementation of $M/M/c$ described in Section~\ref{sec-perf-eval} satisfies Definition~\ref{def-class}.

\begin{remark}
\label{rmk-exch-graph}
The policies in Section~\ref{sec-examples} are exchangeable if ties are broken uniformly at random
within queues. 
Definition~\ref{def-class} allows policies which break symmetry between queues and are not exchangeable.
For instance, the network may have a directed graph topology and the policies
in Section~\ref{sec-examples} may be implemented on the neighborhood of the uniformly chosen queue.
See~\cite{Turner1998}, \cite[Sect.~6]{vanderBoor2022}, \cite{Budhiraja2019}, \emph{e.g.}
\end{remark}

Most if not all asymptotic results are limited to policies which are exchangeable, or at least have an infinitely growing
number of possible choices becoming exchangeable in the limit as in~\cite[Sect.~6]{vanderBoor2022}
and \cite{Budhiraja2019}. Non-asymptotic performance evaluation may also be required.
Simulation studies such as~\cite{Turner1998} usually rely on waiting long lengths of time and 
hoping that equilibrium has been reached without theoretical guidelines.

The present paper will provide theoretical results allowing perfect simulation
for general policies. The special case of exchangeable networks will allow more efficient codes and will be 
investigated specifically.

\section{Preorder and coupling}\label{sec-coupl-preord}

\subsection{Preorder and order up to permutation of coordinates}\label{sec-sto-order}

Let $x\triangleq (x^i)_{1\le i \le c}$ in $\N_0^c$ be a network state and $n$ be in $\N_0$. Let
\begin{equation}
\label{alpha-beta}
\alpha_n(x) \triangleq \sum_{i=1}^c \ind_{\{x^i > n\}}\,,
\qquad
\beta_n(x) \triangleq \sum_{i=1}^c (x^i - n)^+ = \sum_{k\ge n}\alpha_k(x) \,,
\end{equation}
denote respectively the number of queues containing more than $n$ tasks and the number of these excess tasks.
The $\beta_n$ are natural quality of service (QoS) indicators, 
in particular $\beta_0$ counts the total number of customers and for the FCFS discipline
$\beta_n$ counts the number of tasks which have at least $n$ tasks in front of them.

\begin{theorem}
\label{thm-order}
A preorder $\preccurlyeq$ between network states $x$ and $y$ in $\N_0^c$ is defined by
\[
x \preccurlyeq y\iff \beta_n(x) \le \beta_n(y)\,,\; \forall n\in\N_0\,.
\]
The indifference relation $\sim$ defined by $x \sim y$ if and only if $x \preccurlyeq y$ and $y \preccurlyeq x$
satisfies
\[
x \sim y \iff  \exists \sigma\in\mathfrak{S}_c : (x^{\sigma(1)},\dots, x^{\sigma(c)}) = (y^1,\dots,y^c)
\]
and thus $\preccurlyeq$ defines a partial order up to permutation of the coordinates.
Moreover
\[
x \preccurlyeq y\iff \sum_{i=1}^c f(x^i) \le \sum_{i=1}^c f(y^i) \text{ for any non-decreasing convex } f:\R_+ \to \R\,.
\]
\end{theorem}

\begin{proof}
Clearly $\preccurlyeq$ is reflexive and transitive and is thus a preorder, and
$x \sim y$ if and only if $\beta_n(x)  =\beta_n(y) $ for all $n$ in $\N_0$. In particular 
\[
\max\{n\in\N_0 : \beta_n(x) >0\} = \max\{n\in\N_0 : \beta_n(y) >0\} \triangleq m\,,
\quad
\beta_m(x) = \beta_m(y)\,,
\]
which imply that $x$ and $y$ have same number of coordinates with the same maximal size $m+1$.
Inductively by considering $\beta_{m -1}$, \dots\,, $\beta_0$ 
we deduce that $x$ and $y$ have same number of coordinates of 
every size greater than or equal to $1$, and hence are equal up to a permutation of coordinates.

For the last equivalence statement, sufficiency is a special case. 
To prove  necessity, we may and do assume that $f(0)=0$. Then $f$ is continuous on $\R$ and has a 
right-hand derivative $f'_+$ which is non-negative, non-decreasing, right-continuous, and such that
$f(x)  = \int_0^x f'_+(y)\,\mathrm{d}y$, see~\cite[Thms~24.1--2]{Rockafellar1970}.
Integration by parts yields that
\[
f(x)  = xf'_+(x)  - \int_{(0,x]} y\,\mathrm{d}f'_+(y)
 = x f'_+(0)  + \int_{(0,\infty)} (x-y)^+ \,\mathrm{d}f'_+(y)\,.
\]
Thus $f$ is in the positive cone generated by the 
$x\in\R_+ \mapsto (x-y)^+$ for $y\ge0$.
\end{proof}

This preorder is related to the univariate increasing convex order~\cite{Stoyan1983}, 
\cite[Chap.~4]{Shaked-Shantikumar2007}. Its quotient order by $\sim$ is strictly weaker than the product order 
up to permutations of the coordinates, for instance $(1,1)\preccurlyeq(0,2)$.

\subsection{Coupling a load balancing network with the UR network}\label{sec-coupling}

We shall construct a coupling of the queue length processes $X\triangleq (X^1,\dots, X^c)$  of a generic load balancing
network and $Y\triangleq (Y^1,\dots, Y^c)$ of a UR network (Definition~\ref{def-class}) which preserves 
the preorder $\preccurlyeq$. 
A similar coupling is used  in~\cite[Thm~4]{Turner1998} and~\cite[Sect.~4.1]{Graham2000} to compare 
networks implementing JSQ($d$) for different  $d\ge1$.

The coupling will require to rearrange the coordinates of $X$ and of $Y$ in non-decreasing order.
Ties between queues of equal length must be broken using a $\mathfrak{S}_c$-valued process
$(\sigma(t))_{t\in\R_+}$ with nice properties satisfying 
\begin{equation}\label{eq-sigma-tau}
X^{\sigma(t)(1)}(t) \le \dots \le X^{\sigma(t)(c)}(t)\,,
\end{equation}
and similarly for $Y$, with permutation process denoted by $(\tau(t))_{t\in\R_+}$. 

For the sake of simplicity, our exposition will assume that ties are broken by ascending queue index, so that 
$\sigma(t)$ is a deterministic function of $X(t)$. 
Section~\ref{sec-pract-man} will address some practical simulation issues.
It is not efficient to sort $X(t)$ at each update, and $\sigma(t)$ should also be stored and updated. 
The updates of  $\sigma(t)$ when ties are broken by ascending queue index are inefficient, and
we describe there another tie-breaking rule with efficient updates to which our exposition can readily be extended.

The coupling starts by drawing $X(0)$ and $Y(0)$ from a suitable joint law.

The arrivals in both networks follow the same $(\Fc_t)_{t\in\R_+}$-Poisson process of intensity~$\lambda$ 
marked by uniform draws in $\{1,\dots,c\}$. Let $t$ and $k$ be an arrival instant and its mark.
The generic policy takes as its uniformly chosen queue the $k$-th in the non-decreasing ordering, of 
index $\sigma(t-)(k)$, and determines a queue~$i$ such that 
\[
X^i(t-) \le X^{\sigma(t-)(k)}(t-) \,, \quad X(t) = X(t-) + e_i \,.
\]
The UR network does likewise with $i =\tau(t-)(k)$.

The potential service completions in both networks follow the same $(\Fc_t)_{t\in\R_+}$-Poisson process of intensity $c$ 
marked by uniform draws in $\{1,\dots,c\}$.  Let $t$ and $k$ be a potential service completion instant and its mark. 
The generic network has a service completion in the $k$-th queue in the non-decreasing ordering, of index 
$\sigma(t-)(k)$, if and only if it is non-empty.  In this case its policy determines a queue~$i$ such that
\[
X^i(t-)  \ge X^{\sigma(t-)(k)}(t-) \ge1\,, \quad X(t) = X(t-) - e_i \,.
\]      
The UR network does likewise with $i =\tau(t-)(k)$.

\begin{theorem}\label{thm-coupling}
This construction yields a coupling between the queue length process $X\triangleq (X^1,\dots, X^c)$ of a generic 
load balancing network and $Y\triangleq (Y^1,\dots, Y^c)$ of the UR network. Moreover,
\[
X(0) \preccurlyeq Y(0)  \implies X(t) \preccurlyeq Y(t)\,, \; \forall t\in\R_+\,.
\]
\end{theorem}

\begin{proof}
Classic arguments about Poisson process splitting  and about uniform draws and conditioning prove that this is indeed a coupling. 

Assume that $X(0) \preccurlyeq Y(0)$.
Let $t$ be either an arrival or a potential service completion instant.
In a proof by contradiction, assume that
\begin{equation}
\label{eq-proof-contra}
\forall n\in\N_0\,,\; \beta_n(X(t-)) \le \beta_n(Y(t-)) \,,
\quad 
\exists p\in\N_0\,,\; \beta_p(X(t))  > \beta_p(Y(t)) \,.
\end{equation}
This implies that
\begin{equation}
\label{eq-equal-before}
\beta_p(X(t-))  = \beta_p(Y(t-))\,.
\end{equation}

First case: $t$ is an arrival instant with mark $k$.
The generic network allocates the task to queue~$i$ and the UR network to queue~$\tau(t-)(k)$ in a way that  
\begin{equation}
\label{eq-queue lengths-arr}
l \triangleq  X^i(t-) \le X^{\sigma(t-)(k)}(t-)\,,
\qquad
m \triangleq Y^{\tau(t-)(k)}(t-)\,.
\end{equation}
Since $c-\alpha_n$ is the number of coordinates of size at most $n$, the non-decreasing orderings
imply that $c - \alpha_{l-1}(X(t-))   < k$ and  $k \le  c- \alpha_m(Y(t-))$ and thus that
\[
\alpha_m(Y(t-))  < \alpha_{l-1}(X(t-)) \,.
\]
Moreover \eqref{alpha-beta}, \eqref{eq-proof-contra}, \eqref{eq-equal-before}, and \eqref{eq-queue lengths-arr} imply that 
\[
m +1 \le p \le l\,.
\]
Then $p\ge1$, \eqref{alpha-beta} yields that 
\begin{align*}
\beta_{p-1}(X(t-)) &= \alpha_{p-1}(X(t-))  + \beta_p(X(t-))\,,
\\
\beta_{p-1}(Y(t-)) &= \alpha_{p-1}(Y(t-))  + \beta_p(Y(t-))\,,
\end{align*}
and \eqref{eq-proof-contra} and \eqref{eq-equal-before} imply that
\[
\alpha_{p-1}(X(t-)) \le \alpha_{p-1}(Y(t-))\,.
\]
We conclude that
\[
\alpha_{p-1}(X(t-)) \le \alpha_{p-1}(Y(t-)) \le \alpha_m(Y(t-))  < \alpha_{l-1}(X(t-)) \le \alpha_{p-1}(X(t-))
\]
which is a contradiction. Therefore, \eqref{eq-proof-contra} is not true.

Second case: $t$ is a potential service completion instant with mark $k$.
The generic network undergoes an effective service completion in the queue of index $\sigma(t-)(k)$  
if and only if $X^{\sigma(t-)(k)}(t-) \ge1$. 
In this case let $i$ be the index of the queue of which the length will be reduced and else let $i = \sigma(t-)(k)$. 
The UR network reduces by one the length of queue~$\tau(t-)(k)$ if it is not zero. In all cases  
\begin{equation}
\label{eq-queue lengths-dep}
l \triangleq   X^i(t-)  \ge X^{\sigma(t-)(k)}(t-)\,,
\qquad
m \triangleq Y^{\tau(t-)(k)}(t-)\,.
\end{equation}
Since $c-\alpha_n$ is the number of coordinates of size at most $n$, the non-decreasing orderings
imply that $k \le  c- \alpha_l(X(t-))$ and $c - \alpha_{m-1}(Y(t-))   < k$ and thus that
\[
\alpha_l(X(t-)) < \alpha_{m-1}(Y(t-))\,.
\]
Moreover \eqref{alpha-beta}, \eqref{eq-proof-contra}, \eqref{eq-equal-before}, and \eqref{eq-queue lengths-arr} imply that 
\[
l \le p \le m-1\,.
\]
Then~\eqref{alpha-beta} yields that 
\begin{align*}
\beta_p(X(t-)) &= \alpha_p(X(t-))  + \beta_{p+1}(X(t-))\,,
\\
\beta_p(Y(t-)) &= \alpha_p(Y(t-))  + \beta_{p+1}(Y(t-))\,,
\end{align*}
and \eqref{eq-proof-contra} and \eqref{eq-equal-before} imply that
\[
\alpha_p(Y(t-))\le \alpha_p(X(t-)) \,.
\]
We conclude that
\[
\alpha_p(Y(t-))\le \alpha_p(X(t-)) \le \alpha_l(X(t-)) < \alpha_{m-1}(Y(t-))  \le \alpha_p(Y(t-))
\]
which is a contradiction. Therefore, again \eqref{eq-proof-contra} is not true.

Applying the falsity of \eqref{eq-proof-contra} to the successive arrival and potential service completion instants $t$ concludes the proof.
\end{proof}

This yields the following important stability result.

\begin{theorem}
\label{thm-rec-pos}
Let $X$ and $Y$ be the coupled queue length processes of a generic load balancing network and of a UR network.
The empty state $0$ is positive recurrent for $X$ if and only if $\lambda<c$, and then $X$ has a unique
invariant law $\pi$, else there is none.
The state $(0,0)$ is positive recurrent for $(X,Y)$ if and only if $\lambda<c$
and then $(X,Y)$ has a unique invariant law $\kappa$ with marginals $\pi$ and $\gamma^{\otimes c}$,
see~\eqref{geom-law},
and else there is none. If $\lambda<c$ then the two-sided stationary version $(X(t),Y(t))_{t\in\R}$
of $(X,Y)$ satisfies  
\[
X(t) \preccurlyeq Y(t)\,, \quad \forall t\in\R\,.
\]
\end{theorem}
\begin{proof}
If  $\lambda<c$ then the queue length process $Y$ of  the UR network is positive recurrent
since it is irreducible and has $\gamma^{\otimes c}$ as an invariant law, which is then unique.
Theorem~\ref{thm-coupling} yields that if  $(X(0),Y(0)) = (0,0)$ then 
$X(t) \preccurlyeq Y(t)$ and in particular $X(t) =0$ if $Y(t) =0$. 
Therefore $(0,0)$ is positive recurrent for $(X,Y)$ and thus~$0$ is positive recurrent for $X$.
Since $(0,0)$ is reachable from all states for $(X,Y)$, the complement of its recurrence class
is constituted of transient states bearing mass zero for any invariant law, and by restriction
to this recurrence class (necessarily irreducible) there exists a unique invariant law for  $(X,Y)$. Likewise for~$X$.
By comparison with a $M/M/c$ queue, $X$ is at best null recurrent if $\lambda = c$ 
and transient if $\lambda> c$.
The result on the two-sided stationary version $(X(t),Y(t))_{t\in\R}$ follows from Palm theory
using the hitting times of $(0,0)$ as regeneration times splitting the process into cycles, 
see~\cite[Sect.~8.3]{Thorisson2000}.
Indeed, the i.i.d.\ cycles start with $0 = X(0) \preccurlyeq Y(0) =0$ and thus satisfy
$X(t) \preccurlyeq Y(t)$ during their lifetimes, hence the length-biased cycle straddling time~$0$ 
must also satisfy $X(t) \preccurlyeq Y(t)$ during its lifetime.
\end{proof}

A load balancing parallel server queueing network is irreducible:
any state leads to the empty state, and in order to attain any given state from the empty state we may
load all queues beyond the maximal length of the given state and then serve the excess tasks. 
A network implementing work stealing may well not be irreducible.

In the sequel we assume that the stability condition $\lambda <c$ holds true.

\section{Considerations for effective simulation}
\label{sec-sim-sym}

\subsection{Sampling from the invariant law of the UR network}
\label{rem-geom}

The perfect simulation algorithms will require to draw repeatedly from the invariant law $\gamma^{\otimes c}$
of the queue length process $Y$ of the uniform routing (UR) network, 
where $\gamma$ is the geometric law~\eqref{geom-law} under the stability condition $\lambda <c$.

A simple option for each draw is to simulate $c$ independent draws from $\gamma$. 
A bounded time inversion method using truncated exponential random variables 
is given in Devroye~\cite[Sect. X.2.2]{Devroye-1986}. It is noted in an exercise that an accuracy problem
in computing $\log (1-p)$ could arise if the success probability  $p$ is small, and an expansion method is proposed, but 
here $p = 1-\lambda/c$ and we can use $\log (\lambda/c )= \log \lambda - \log c $.

Another option is to first draw the total number of tasks from the negative binomial law  
$\mathrm{NB}(c,1-\lambda/c)$ and then draw from a multinomial to attribute tasks to the $c$ queues.
Draws from $\mathrm{NB}(c,1-\lambda/c)$ can be performed by inverting the c.d.f., which is a regularized incomplete 
beta function, or by other uniformly fast methods as in~\cite[Sect.~X.1.2 Example~1.5, Sect.~X.4.7]{Devroye-1986}.
Methods for drawing from a multinomial with mean durations which are $O(c)$ uniformly in the 
total number of tasks exist, such as~\cite[Sect.~XI.1.4 Example~1.5]{Devroye-1986}.

\subsection{Exchangeable load balancing policies}
\label{sec-exch}

This subsection is devoted to exchangeable  load balancing policies, see Definition~\ref{def-class}. 
A reduced representation of $X = (X(t))_{t\in\R_+}$ and of $Y = (Y(t))_{t\in\R_+}$ ``up to instantaneous
permutations of the coordinates'' is then adequate to study the invariant law. 
The exchangeability  assumption is equivalent to assuming that these reduced representations are themselves Markov. 
An added benefit is that the preorder $\preccurlyeq$ of Theorem~\ref{thm-order} induces an order in this context.

For $x\triangleq (x^i)_{1\le i\le c}$ in $\N_0^c$ let $\acute{x} \triangleq (\acute{x}^i)_{1\le i\le c}$ in $\N_0^c$ 
be the vector of its coordinates set in non-decreasing order and
\begin{equation*}
h_n(x) \triangleq \sum_{i=1}^c \ind_{\{x^i \le n\}} = c- \alpha_n(x) 
\end{equation*}
be the number of its coordinates not exceeding $n$ in $\N_0$. Setting $h_{-1}(x)= 0$,
\begin{equation}\label{h-and-x}
\acute{x}^i =n \iff h_{n-1}(x)+1 \le i \le h_n(x) \,,
\qquad
h_n(x) = h_n(\acute{x})\,.
\end{equation}
Note that $x$ can be sorted into $\acute{x}$ in $O(c \log c)$ time, and that $x\sim y \Leftrightarrow \acute{x} = \acute{y}$, see Theorem~\ref{thm-order}.
Let $\acute{X}(t)$ and $\acute{Y}(t)$ be the non-decreasing reorderings of $X(t)$ and $Y(t)$.

The processes $\acute{X} \triangleq (\acute{X}(t))_{t\in\R_+}$ and $\acute{Y} \triangleq (\acute{Y}(t))_{t\in\R_+}$ are reduced representations of $X$ and $Y$ which are $(\Fc_t)_{t\in\R_+}$-Markov and evolve as follows.
\begin{enumerate}
\item\label{arrival-update-exch}
Let $t>0$ be an arrival instant with mark $k$. The task is allocated according to policy to a queue of 
length $n\le \acute{X}^{k}(t-)$, and the single coordinate update is 
\[
\acute{X}^{h_n(\acute{X}(t-))}(t)  = \acute{X}^{h_n(\acute{X}(t-))}(t-)+1= n+1\,.
\]
\item\label{depart-update-exch}
Let $t>0$ be a potential service completion instant with mark $k$.
It is effective if and only if $\acute{X}^{k}(t-)\ge1$, and then a task is removed according to policy
from a queue of length $n\ge \acute{X}^{k}(t-) \ge1$, and the single coordinate update is 
\[
\acute{X}^{h_{n-1}(\acute{X}(t-))+1}(t)  = \acute{X}^{h_{n-1}(\acute{X}(t-))+1}(t-) -1= n-1\,.
\]
\end{enumerate}
The determination of the queue of length $n$ by the policy is Markov due to exchangeability.
For instance, JSQ($d$) may choose $k_1 = k\,, k_2\,,\dots \,, k_d$ uniformly and 
take $n = \acute{X}^{\min\{k_1,\dots,k_d\}}(t-)$
in \ref{arrival-update-exch}.  
The special case $\acute{Y}$ uses $n=\acute{Y}^{k}(t-) $
in \ref{arrival-update-exch} and \ref{depart-update-exch}.

An actual code can be managed efficiently by storing $\acute{X}(t)$ and $\acute{Y}(t)$ each in an array of $c$ cells and using $\eqref{h-and-x}$ at each step. 
Indeed, $h_{n-1}(\acute{X}(t))+1$ and $h_n(\acute{X}(t))$ are the indices of the first and last cells storing  the value $n$ in $\acute{X}(t)$ and can  be found in $O(\log c)$ time by binary search of the sorted array $\acute{X}(t)$.
Usually the load balancing policy provides an index $j$ such that $\acute{X}^j(t-) = n$, and the search for $h_{n-1}(\acute{X}(t))+1$ [resp.\ $h_n(\acute{X}(t))$] can be limited to the cells before 
[resp.\ after] the $j$-th.

\begin{remark}
\label{rem-exch-or-not}
We shall develop perfect simulation methods for $\pi$ which use the coupling $(X,Y)$,
and actually simply $(X,\acute{Y})$. 
The arguments are readily adapted for exchangeable policies using $(\acute{X},\acute{Y})$, providing 
perfect simulation methods for the invariant law $\acute{\pi}$ of $\acute{X}$.  
If $f$ is symmetric then $\int f d\pi = \int f d\acute{\pi}$, \emph{e.g.}, the $\beta_n$ are of this form,
else a uniform random permutation of the coordinates of a sample from $\acute{\pi}$, independent of the rest, 
yields a sample from $\pi$. 
\end{remark}

\subsection{General load balancing policies}
\label{sec-pract-man}

Certain load balancing policies  need consider the specific indices of the queues. 
For instance the network may have a directed graph structure and may implement policies such as those in  Section~\ref{sec-examples} only on the graph neighborhood of the uniformly chosen queue.
Then $X$ is not exchangeable and $\acute{X}$ is not Markov. 

Sorting $X(t)$ at each update would take $O(c \log c)$ time.
It is usually much more efficient to store and update
\begin{equation}\label{eq-store-x-sig}
X(t) \triangleq (X^i(t))_{1\le i \le c}\,,
\qquad
\sigma(t) \triangleq (\sigma(t)(i))_{1\le i \le c}\,,
\end{equation} 
satisfying~\eqref{eq-sigma-tau}, in two arrays of $c$ cells, and use when needed that
\begin{equation}\label{eq-acute-sig}
\acute{X}^i(t) = X^{\sigma(t)(i)}(t)\,,
\qquad
1\le i\le c\,,
\end{equation}
for instance to find $h_n(X(t))=h_n(\acute{X}(t))$ by binary search of $\acute{X}(t)$, see Section~\ref{sec-exch}, which uses only the elements for comparison and takes $O(\log c)$ time.

For the sake of simplicity, the tie-breaking rule used for $\sigma(t)$ in the exposition was by 
ascending queue indices, in which case every update of $\sigma(t)$ may require up to $c$ memory swaps. 
The following tie-breaking rule has efficient updates swapping at most the contents of two memory cells. 
The results of the exposition can readily be extended to it by considering the 
$(\Fc_t)_{t\in\R_+}$-Markov process $(X(t),\sigma(t))_{t\in\R_+}$.

\begin{enumerate}\item\label{arrival-update-gen}
Let $t>0$ be an arrival instant with mark $k$. The task is allocated to a queue of index~$i = \sigma(t-)(j)$ 
and length $n\le X^{\sigma(t-)(k)}(t-)$ and the updates are
\begin{gather*}
X^i(t) = X^i(t-) +1 = n+1\,,
\\
\sigma(t)(h_n(X(t-)))=\sigma(t-)(j) =i\,,
\quad
\sigma(t)(j) = \sigma(t-)(h_n(X(t-)))\,,
\end{gather*}
\emph{i.e.}, $h_{n-1}(X(t-)) < j\le h_n(X(t-))$ and $\sigma(t)$ is obtained from $\sigma(t-)$ by swapping the index $i$ 
in cell~$j$ with the index in cell $h_n(X(t-))$. This propagates~\eqref{eq-sigma-tau}. 
\item\label{depart-update-gen}
Let $t>0$ be a potential service completion instant with mark $k$.
It is effective if and only if $X^{\sigma(t-)(k)}(t-)\ge1$, and then a task is removed from a queue 
of index~$i = \sigma(t-)(j)$ and length $n\ge X^{\sigma(t-)(k)}(t-) \ge1$ and the only updates are
\begin{gather*}
X^i(t) = X^i(t-) -1 = n-1\,,
\\
\sigma(t)(h_{n-1}(X(t-))+1) =\sigma(t-)(j) =i\,,
\quad
\sigma(t)(j) = \sigma(t-)(h_{n-1}(X(t-))+1)\,.
\end{gather*}
\emph{i.e.}, $h_{n-1}(X(t-)) < j\le h_n(X(t-))$ and $\sigma(t)$ is obtained from $\sigma(t-)$ by swapping the index 
$i$ in cell~$j$ with the index in cell $h_{n-1}(X(t-))+1$. This propagates~\eqref{eq-sigma-tau}. 
\end{enumerate}

The special case $Y$ has $j=k$ in \ref{arrival-update-gen} and 
\ref{depart-update-gen}, but we shall see that we need only simulate $\acute{Y}$ 
as in Section~\ref{sec-exch} in the perfect simulation algorithms.

\subsection{Subordination and embedded Markov chain}\label{sec-emb-mc}

Let $(T_n)_{n\ge1}$ denote the sequence of arrival and potential service completion instants set in increasing order and
\[
\theta_n =
\left\{
\begin{aligned}
& (1,k) &&\text{if $T_n$ is an arrival instant with mark $k$,}
\\
& (-1,k) &&\text{if $T_n$ is an potential service completion instant with mark $k$.}
\end{aligned}
\right.
\]
Classically, $(T_n,\theta_n)_{n\ge1}$ is a $(\mathcal{F}_t)_{t\in\R_+}$-marked Poisson process 
of intensity $\lambda + c$, and the marks $\theta_n$ have common law 
\begin{equation}\label{law-marks}
\biggl(\frac{c}{\lambda +c}\delta_{-1} + \frac{\lambda}{\lambda +c}\delta_1\biggr)
\otimes \frac{1}{c}(\delta_1 + \cdots + \delta_c)
\;\;\text{on}\;\;  \Theta \triangleq \{-1,1\} \times \{1,\dots,c\}\,.
\end{equation}

The Markov process $(X(t),Y(t))_{t\in\R_+}$ is subordinate to
$(T_n)_{n\ge1}$, and thus has same invariant law $\kappa$ as the embedded $(\Fc_{T_n})_{n\in\N_0}$-Markov chain 
$(X_n,Y_n)_{n\in\N_0}$ defined by
\begin{equation*}
(X_0,Y_0) = (X(0),Y(0))\,,
 \qquad
 (X_n,Y_n) = (X(T_n),Y(T_n))\,.
\end{equation*}
Simulating it avoids the simulation of exponential random variables $\mathcal{E}(\lambda+c)$.

We wish to implement the coupling $(X_n,Y_n)_{n\in\N_0}$ using update functions and i.i.d.\ random draws.
For $n\ge1$, $X_n$ is obtained from $X_{n-1}$ and $\theta_n$ according to the load balancing policy
in a way which may involve identically distributed random variables independent of $\Fc_{T_{n-1}}$,
\emph{e.g.}, to break ties between queues of equal length or implement JSQ($d$),
and $Y_n$ is obtained from $Y_{n-1}$ and $\theta_n$ deterministically, as in Section~\ref{sec-coupling}.
We write this
\begin{equation}
\label{update-fct}
 X_n= \phi(X_{n-1},\theta_n,U_n)\,, \quad Y_n =\psi(Y_{n-1},\theta_n)\,,
\end{equation}
using deterministic update functions $\phi:\N_0^c \times \Theta\times [0,1]\mapsto \N_0^c$ and $\psi : \N_0^c \times \Theta\mapsto \N_0^c$ and a $(\Fc_{T_n})_{n\in\N_0}$-adapted sequence $(U_n)_{n\ge1}$ of
i.i.d.\ uniform random variables on $[0,1]$ independent of $(X_0,Y_0, T_n,\theta_n;n\ge1)$. 

If $X$ uses an exchangeable policy as does $Y$, then $(\acute{X}_n)_{n\in\N_0}$ is a $(\mathcal{F}_{T_n})_{n\in\N_0}$-Markov chain itself, see Section~\ref{sec-exch}, and we can write
\begin{equation}\label{update-fct-exch}
 \acute{X}_n= \acute{\phi}(\acute{X}_{n-1},\theta_n,U_n)\,, \quad 
 \acute{Y}_n =\acute{\psi}(\acute{Y}_{n-1},\theta_n)\,.
\end{equation}

In the sequel we use $(X_n,Y_n)_{n\in\N_0}$ and \eqref{update-fct} in three perfect simulation algorithms for the invariant law $\pi$ of the queue length process $X$ of a load balancing network.  
The invariant law $\kappa$ has marginals $\pi$ and 
$\gamma^{\otimes c}$, and  $(Y_n)_{n\in\N_0}$ in equilibrium is reversible  
and will be used as a simulatable dominating Markov chain in this infinite state space.
We shall see we need only simulate $(\acute{Y}_n)_{n\in\N_0}$ using~\eqref{update-fct-exch}
for our purposes. 

For the sake of simplicity, the exposition breaks ties by ascending queue index.
In order to use the efficient tie-breaking rule of Section~\ref{sec-pract-man}, we need only consider 
$(X_n,\sigma_n)= \phi(X_{n-1}, \sigma_{n-1},\theta_n,U_n)$, etc., where  $\sigma_n = \sigma(T_n)$. 

Exchangeable policies require only to simulate $\acute{\pi}$ using $(\acute{X}_n,\acute{Y}_n)_{n\in\N_0}$ 
and~\eqref{update-fct-exch}.

\section{Palm theory and acceptance-rejection}\label{sec-Palm}

A classic result in Palm theory, see~\cite[Chap.~8,~10]{Thorisson2000}, is that a stationary version on $\R$ of a cycle-stationary process can be obtained by length-biasing the law of the cycle-stationary process by the length of the cycle straddling the origin and then setting the origin uniformly in the length of this cycle. 
The trace on $\R_+$ of the cycle straddling the origin is called the delay. 

In~\cite[Sect.~8.6.7]{Thorisson2000} the simulation of the stationary delay is obtained by
the following acceptance-rejection method.
\begin{enumerate}
\item
Simulate the length of a stationary delay.
\item
Simulate independent copies of the cycles, until the length of a cycle exceeds the 
previously simulated stationary delay length.
\item
The part of this cycle from the stationary delay length onward constitutes a simulation of the stationary delay. 
In particular, the state of this cycle at the stationary delay length is a draw from the invariant law.
\end{enumerate}
Moreover, the number of trials is finite, a.s., but has infinite expectation. 
This continuous-time result can be applied to discrete time as in~\cite[Sect.~10.2.6]{Thorisson2000}.

Sigman~\cite{Sigman2012} provides an interesting implementation of this result in discrete time,
and gives it a short direct proof in its Prop.~2.1. We adapt this implementation here as a warm-up 
for illustrative purposes, since the infinite expectation of the number of trials renders it unsuitable for 
actual Monte Carlo simulation.
The situation is simpler than for DomCFTP since the simulations are all in direct time and there is no difficulty
in preserving the coupling by using the same marks $\theta_n$ in~\eqref{update-fct}.

The hitting times of $(0,0)$ by the Markov chain $(X_n,Y_n)_{n\in\N_0}$ are regeneration instants. 
Theorems~\ref{thm-coupling} and~\ref{thm-rec-pos} yield that if $(X_0,Y_0) = (0,0)$ 
or follows the invariant law $\kappa$ then $X_n \preccurlyeq Y_n$ for all $n$ in $\N_0$. Thus in both cases 
these regeneration instants are given by the hitting times of~$0$ of  the Markov chain $(Y_n)_{n\in\N_0}$.
These are the same as those of $(\acute{Y}_n)_{n\in\N_0}$, and we need only simulate the latter.

Hence, the perfect simulation method first obtains the length $L$ of a stationary delay by simulating an instance of 
$(\acute{Y}_n)_{n\in\N_0}$ started at a draw from its invariant law, obtained by sorting in non-decreasing order a draw 
from $\gamma^{\otimes c}$ (Section~\ref{rem-geom}). It then determines the first cycle of length exceeding $L$ 
by simulating instances of $(\acute{Y}_n)_{n\in\N_0}$ started at $0$. 
When this is obtained, the method reuses the marks of this cycle in \eqref{update-fct} to simulate the coupled $X_0, \dots, X_L$, and outputs $X_L$ as a perfect draw from $\pi$. 

For exchangeable load balancing policies, the same reasoning is valid to simulate the invariant law $\acute{\pi}$ of 
$\acute{X}$ by using \eqref{update-fct-exch} in place of \eqref{update-fct}; see Remark~\ref{rem-exch-or-not}.

We recapitulate this in the following perfect simulation algorithm, written in mathematical fashion and not in pseudocode.
It has infinite expected running time and is not to be implemented in practice.

\begin{algorithm}[H] 
\caption{Acceptance-rejection algorithm}\label{alg-acc-rej}
\renewcommand\algorithmicdo{} 
\algnotext{EndFor}
\begin{algorithmic}[1]
\State draw $Y_0$ from $\gamma^{\otimes c}$ and sort into $\acute{Y}_0$
\Comment see Section~\ref{rem-geom} 
\For{$n=1$ to $L \triangleq \inf\{n\ge 1: \acute{Y}_n =0\}$}\label{a1-delay-iter}
	\State draw  $\acute{Y}_n$ using $\acute{Y}_{n-1}$
	\Comment using \eqref{update-fct-exch} or otherwise
\EndFor
\For{$k=1$ to $K\triangleq \inf\{k\ge 1 : L^{(k)} = L\}$}\label{a1-coeur}
	\State set $\acute{Y}^{(k)}_0=0$
	\For{$n=1$ to $L^{(k)}\triangleq \inf\{n\ge 1 : \acute{Y}^{(k)}_n =0\} \wedge L$}
		\State draw $\theta^{(k)}_n$ from \eqref{law-marks} 
		\State set $\acute{Y}^{(k)}_n = \acute{\psi}(\acute{Y}^{(k)}_{n-1},\theta^{(k)}_n)$ 
	\EndFor
\EndFor
\State set $X^{(K)}_0 = 0$ 
\For{$n=1$ to $L$}\label{a1-perf-sim-iter}
	\State draw $U_n$ uniformly in $[0,1]$ 
	\Comment for possible random choices
	\State set $X^{(K)}_n=\phi(X^{(K)}_{n-1},\theta^{(K)}_n,U_n)$
	\Comment the $\theta^{(K)}_n$ of Loop~\ref{a1-coeur} 
\EndFor
\State\Return $X^{(K)}_L$
\end{algorithmic}
\end{algorithm}

\begin{theorem}
Let $\lambda < c$, and $\pi$ be the invariant law of the queue length process $X$ of a load balancing network
of Definition~\ref{def-class}. 
Algorithm~\ref{alg-acc-rej} uses simulatable draws, terminates in an a.s.\ finite time with infinite expectation,
and outputs a perfect draw from $\pi$, \emph{i.e.}, 
$\P(K<\infty)=1$, $\E(K) = \infty$, and  $X^{(K)}_L$ has law $\pi$.
If the load balancing policy is exchangeable, Algorithm~\ref{alg-acc-rej} with $X^{(K)}$ and \eqref{update-fct} 
replaced by $\acute{X}^{(K)}$ and \eqref{update-fct-exch} in Loop~\ref{a1-perf-sim-iter} 
outputs a perfect draw from the invariant law $\acute{\pi}$ of $\acute{X}$; see Remark~\ref{rem-exch-or-not}.
\end{theorem}
\begin{proof}
The proof is given in the preceding discussion.
\end{proof}

\section{Dominated Coupling From The Past (DomCFTP)}\label{sec-DomCFTP}

\subsection{Background}

Coupling From The Past (CFTP) was introduced by Propp and Wilson~\cite{Propp1996}.
It concerns a Markov chain, on a finite state space with an order and a least and a largest state, 
which can be constructed using an update function preserving this order and an i.i.d.\ sequence of random variables.
This seminal paper provided impressive applications to the Ising model close to criticality
and much general insight.

There is first a \emph{back-off} strategy.
Instances of the Markov chain are repeatedly started from the least and the largest states
further and further back in the past before time~$0$.
All these instances are coupled by always using the same i.i.d. random variables driving 
the updates from times $n-1$ to $n$.
This is repeated until \emph{coalescence} happens: the two instances have coalesced by time~$0$, \emph{i.e.},
they have taken the same value and necessarily remained equal ever since. 
Due to the preservation of the order, all instances of the Markov chain started arbitrarily
at the corresponding time back in the past must have coalesced, hence
the common value obtained at time~$0$ must coincide with the value of the Markov chain in equilibrium 
at time~$0$ and constitute a perfect draw from the invariant law.

Dominated CFTP (DomCFTP) extends this to infinite state spaces such as $\N_0^c$,
see~\cite{Kendall1998,Kendall2005,Kendall-Moller2000,Thonnes2000,Connor-Kendall2015,Huber2016}.
DomCFTP simulates backward in time a dominating Markov chain in order to obtain lower and upper bounds
for the Markov chain in equilibrium, coupled in the subsequent forward simulation together with the
dominating chain.

\subsection{DomCFTP at the empty state}

A novelty here is the use of the preorder of Section~\ref{sec-sto-order} in the truly multi-dimensional setting of queueing networks. 
The  order ``up to coordinate permutation'' on the quotient set for the indifference relation $\sim$ is strictly weaker than the classic product order, and the latter is not in general preserved by the coupling of  Section~\ref{sec-coupling}.

Perfect simulation methods for non-exchangeable polices are of great interest, see Remark~\ref{rmk-exch-graph}, and require the use of the preorder with its lack of anti-symmetry.
If the policy is exchangeable then we can use the reduced representation of Section~\ref{sec-exch} on which the preorder induces an order; we shall see a use of this in Section~\ref{sec-sandwich}.

A way to implement DomCFTP is to wait for the dominating chain simulated backward in time to hit $0$.
At this random time, the chain in equilibrium must be at state $0$ (coalescence is achieved without anti-symmetry) and can be simulated forward in time up to time~$0$, at which time its state constitutes a perfect draw from $\pi$.

A delicate point is that the Markov chains must be coupled together through the
backward and forward simulations by using the same driving i.i.d.\ random variables.
We deal with this by considering the two-sided stationary Markov chain
\[
(X_n,\acute{Y}_n,\theta_{n+1})_{n\in\Z}\,.
\]

\begin{theorem}
\label{theta-from-y}
The two-sided stationary Markov chain $(\acute{Y}_n,\theta_{n+1})_{n\in\Z}$ in equilibrium can be perfectly simulated
backward in time from time $0$.  Using reversibility, $(\acute{Y}_n)_{n\in\Z_-}$ can be simulated by
drawing $Y_0$ of law $\gamma^{\otimes c}$, sorting it into $\acute{Y}_0\triangleq \acute{Y}_0'$,
simulating a copy $(\acute{Y}'_n)_{n\in\N_0}$ of $(\acute{Y}_n)_{n\in\N_0}$,
and setting $\acute{Y}_n = \acute{Y}'_{-n}$ for $n\le 0$. 
Moreover,
\begin{itemize}
\item
if $\acute{Y}_{n+1} = y+e_i$ and $\acute{Y}_n =y$ and $y^i = m$ for $y$ in $\N_0^c$ and $i$ in $\{1,\dots,c\}$,
then 
\[
\theta_{n+1} = (1,K) \text{ with } K \text{ uniform in }
\{h_{m-1}(y) +1\,, \dots\,, h_m(y) = i\}\neq\emptyset\,,
\]
\item
if $\acute{Y}_{n+1} = y-e_i$ and $\acute{Y}_n =y$ and $y^i = m\ge1$ for $y$ in $\N_0^c$ and $i$ in $\{1,\dots,c\}$,
then 
\[
\theta_{n+1} = (-1,K) \text{ with } K \text{ uniform in }
\{h_{m-1}(y) +1 = i\,, \dots\,, h_m(y)\}\neq\emptyset\,,
\]
\item
if $\acute{Y}_{n+1} = y$ and $\acute{Y}_n = y$ for $y$ in $\N_0^c$, then 
\[
\theta_{n+1} = (-1,K) \text{ with } K \text{ uniform in }
\{1\,, \dots\,, h_0(y)\}\neq\emptyset\,.
\]
\end{itemize}
\end{theorem}

\begin{proof}
The statement about reversibility is standard. The statements about uniform laws follow from
the fact that the $\theta_n$ have law~\eqref{law-marks} and~\eqref{h-and-x} and conditioning. 
For completeness, we provide a direct proof of the last case, the others being similar. 
If $k$ is in $\{1\,, \dots\,, h_0(y)\}\neq \emptyset$ then $y^k=0$ and
\begin{align*}
\P(\theta_{n+1}  = (-1,k) \,|\, \acute{Y}_{n+1} =y, \acute{Y}_n = y) 
&\triangleq \frac{\P(\acute{Y}_n = y, \theta_{n+1}  = (-1,k),\acute{Y}_{n+1} =y)}{\P(\acute{Y}_n = y, \acute{Y}_{n+1} =y)}
\\
&= \frac{\P(\acute{Y}_n = y, \theta_{n+1}  = (-1,k))}{\P(\acute{Y}_n = y, \acute{Y}_{n+1} =y)}
\\
&= \frac{\P( \acute{Y}_n = y) \,\P(\theta_{n+1}  = (-1,k))}{\P(\acute{Y}_n = y) \,\P(\acute{Y}_{n+1} =y \,|\, \acute{Y}_n = y)}
\end{align*}
and $\P(\theta_{n+1}  = (-1,k)) = \frac{1}{c(\lambda +c)}$ and 
$\P(\acute{Y}_{n+1} =y \,|\, \acute{Y}_n = y) = \frac{h_0(y)}{c(\lambda +c)}$ yield that
\[
\P(\theta_{n+1}  = (-1,k) \,|\, \acute{Y}_{n+1} =y, \acute{Y}_n = y) = \frac{1}{h_0(y)}\,.
\qedhere
\]
\end{proof}

We should consider $\acute{Y}_n$ instead of  $Y_n$ for practical simulation issues even though the $\tau_n$
are not needed, for instance in order to allow to determine the $h_m$ by binary search.
A possibility for simulating $(\acute{Y}_n)_{n\in\N_0}$ 
is to use \eqref{update-fct-exch} as  $\acute{Y}'_n = \acute{\psi}(\acute{Y}'_{n-1}, \theta'_n)$ for $n\ge1$ with 
$\theta'_n$ i.i.d.\ of  law \eqref{law-marks}, which results in
\begin{equation}\label{back-y}
Y_0 \sim \gamma^{\otimes c}\,,\qquad
\acute{Y}_n = \psi(\acute{Y}_{n+1}, \theta'_{-n})\,,
\quad
n\le -1\,.
\end{equation}
These $\theta'_n$ are only related indirectly to the $\theta_n$ in Theorem~\ref{theta-from-y}.

The DomCFTP method uses Theorem~\ref{theta-from-y}. 
It simulates backwards in time $(\acute{Y}_n)_{n\in\Z_-}$ until the backward stopping time
\begin{equation}
\label{eq-empty-back}
N\triangleq \sup\{n\le 0 : \acute{Y}_n =0\} = \sup\{n\le 0 : Y_n =0\} \le 0
\end{equation}
which will be proved to be finite, a.s., 
Theorem~\ref{thm-rec-pos} yields that $0 \preccurlyeq X_N \preccurlyeq Y_N =0$
and hence that $X_N =0$ (even though $\preccurlyeq$ is only a preorder).
Moreover, $\theta_{n+1}$ can be simulated once $\acute{Y}_n$ has been simulated from $\acute{Y}_{n+1}$ for $n\le-1$.
Then $\theta_{N+1}$, \dots\,, $\theta_0$ and~\eqref{update-fct} allow to simulate
$X_{N+1}$, \dots\,, $X_0$, and $X_0$ is a perfect draw from $\pi$.

This results in the following perfect simulation algorithm, written in mathematical fashion and not in pseudocode.
An actual implementation should follow Section~\ref{sec-sim-sym}.

\begin{algorithm}[H]
\caption{DomCFTP at the empty state}\label{alg-dcftp-empty}
\begin{algorithmic}[1]
\renewcommand\algorithmicdo{}
\algnotext{EndFor}
\State draw  $Y_0$ from $\gamma^{\otimes c}$ and sort into $\acute{Y}_0$
\Comment see Section~\ref{rem-geom}
\For{$n = -1$ down to $N\triangleq \sup\{m\le 0 : \acute{Y}_m =0\}$}\label{a2-theta}
\Comment not executed if $\acute{Y}_0=0$
	\State draw $\acute{Y}_n$ using $\acute{Y}_{n+1}$
	\Comment see Theorem~\ref{theta-from-y} and \eqref{back-y}
	\State draw $\theta_{n+1}$ using $\acute{Y}_{n+1}$ and $\acute{Y}_n$
	\Comment see Theorem~\ref{theta-from-y}
\EndFor
\State set $X_N = 0$
\For{$n=N+1$ up to $0$}\label{a2-perf-sim-iter}
\Comment not executed if $\acute{Y}_0=0$
	\State draw $U_n$ uniformly in $[0,1]$ 
	\Comment for possible random choices
	\State set $X_n = \phi(X_{n-1},\theta_n,U_n)$
	\Comment the $\theta_n$ obtained in Loop~\ref{a2-theta}
\EndFor
\State\Return $X_0$
\end{algorithmic}
\end{algorithm}
 
\begin{theorem}\label{DCFTP-empty}
Let $\lambda < c$, and $\pi$ be the invariant law of the queue length process $X$ of a load balancing network
of Definition~\ref{def-class}. 
Algorithm~\ref{alg-dcftp-empty} uses simulatable draws, terminates after a duration which has 
some exponential moments, and outputs a perfect draw from~$\pi$. 
Specifically, there exists $s>1$ such that $\E(s^{\lvert N\rvert})<\infty$, hence $\E(\lvert N\rvert)<\infty$,
and  $X_0$ has law $\pi$.
If the load balancing policy is exchangeable, Algorithm~\ref{alg-dcftp-empty} with $X$ and \eqref{update-fct} 
replaced by $\acute{X}$ and \eqref{update-fct-exch} in Loop~\ref{a2-perf-sim-iter} 
outputs a perfect draw from the invariant law $\acute{\pi}$ of $\acute{X}$; see Remark~\ref{rem-exch-or-not}.

\end{theorem}

\begin{proof}
The proof follows from the preceding discussion except the statements on $N$. Now, $-N = \lvert N\rvert$
has same law as the hitting time of $0$ by $(Y_n)_{n\ge0}$ in equilibrium.
Classically, when $c=1$ the hitting time of $0$ by $(Y_n)_{n\ge0}$ starting at $0$ has some exponential moments,
and the hitting time of $0$ by $(Y_n)_{n\ge0}$ in equilibrium is a sized biased version of it multiplied by an 
independent uniform random variable on $[0,1]$ and thus also has some exponential moments.
The conclusion follows from \cite[Sect.~II.1.4]{Lindvall1992}, or \cite[Sect.~10.7]{Thorisson2000}
adapted to discrete time, and induction over $c\ge2$.
\end{proof}

\begin{remark}\label{rem-sim-1-paral}
This algorithm allows perfect simulation of the invariant laws of coupled networks with varied load balancing
policies, for instance in order to compare their performances. 
The backward simulation of $(\acute{Y}_0,\dots,\acute{Y}_N)$ and of $(\theta_0,\dots,\theta_{N+1})$ is performed once. 
Then the queue length processes of these networks are simulated forward in time from time $N$ and state~$0$ 
up to time~$0$, possibly in parallel fashion.
\end{remark}

\subsection{DomCFTP with back-off and sandwiching}
\label{sec-sandwich}

In practice we expect $c$ to be at least moderately large and $\lambda$ to be close to $c$. Then 
the time $\lvert N \rvert$ that $(\acute{Y}_n)_{n \in \Z_-}$ takes to reach the empty state $0$ 
and the duration of the simulation 
procedure in Theorem~\ref{DCFTP-empty} are likely to be quite long.
It takes much shorter times for $(\acute{Y}_n)_{n \in \Z_-}$ to reach states in which some queues but not all are empty,
which are those such that a queue in $(X_n)_{n \in \Z_-}$ may finish serving a task while $(\acute{Y}_n)_{n \in \Z_-}$ 
does not change. A corresponding back-off method may achieve a considerable speed-up.

An important issue is that the preorder~$\preccurlyeq$  can be used for these purposes only if the update function $\phi$ preserves the quotient order for the indifference relation $\sim$ in Theorem~\ref{thm-order}. 
This is very close to requiring that the load balancing policy be exchangeable; an uninteresting generalization 
is for instance a policy such as JSQ or JSQ(d) breaking ties by ascending queue index. 
For the sake of clarity, we henceforth focus on exchangeable load balancing networks, see Section~\ref{sec-exch}.

We are going to provide a DomCFTP method for the perfect simulation of the invariant law $\acute{\pi}$ of the 
instantaneous non-decreasing ordering $\acute{X}$ of $X$ using the Markov chain $(\acute{X}_n)_{n \in \Z_-}$. 
It requires to choose a back-off strategy
\begin{equation}
\label{backoff}
-\infty < \cdots <M_k<\dots < M_1 <0
\end{equation}
constituted of stopping times of  the dominating Markov chain $(\acute{Y}_n)_{n\in\Z_-}$ simulated backward 
in time. A classic choice is binary (or doubling) back-off,  in which 
\[
M_k = 2 M_{k-1} = 2^{k-1}  M_1\,, \quad k\ge2\,.
\]
It has some empirical and theoretical backing, see~\cite[Sect.~5.2]{Propp1996}. Usually $M_1$ is deterministic 
but could well be random, such as the first backward time a queue empties for the dominating Markov chain.

Recall that $Y_0$ can be drawn from $\gamma^{\otimes c}$, see Section~\ref{rem-geom}, and rearranged into 
$\acute{Y}_0$. If $Y_0 =0$ then coalescence has been achieved at time $0 \triangleq M_0$ and the 
procedure stops, else we recursively back-off as follows.  Recall Theorem~\ref{theta-from-y}.

For $k\ge 1$, if coalescence has not been achieved yet after simulating forward from time $M_{k-1}$, after which
\[
\acute{Y}_0,\quad \acute{Y}_{-1}, \theta_0,\, \dots\,, \acute{Y}_{M_{k-1}}, \theta_{M_{k-1}+1}, \quad U_{M_{k-1}+1}, \dots\,, U_0
\]
have been determined, then $M_k$ is determined by simulating backward in time 
\[
\acute{Y}_{M_{k-1}-1}, \theta_{M_{k-1}},\, \dots\,,  \acute{Y}_{M_k}, \theta_{M_k+1}\,.
\]
Then we construct forward in times, drawing $U_{M_k+1}, \dots\,, U_{M_{k-1}}$
and reusing $U_{M_{k-1}+1}, \dots\,, U_0$,
\begin{align*}
\check{X}_{M_k}=0\,, 
&&\check{X}_{{M_k}+1} = \acute{\phi}(\check{X}_{M_k}, \theta_{{M_k}+1},U_{{M_k}+1})\,, \,\dots\,,\,
\check{X}_0= \acute{\phi}(\check{X}_1, \theta_0,U_0)\,, 
\\
\hat{X}_{M_k}=\acute{Y}_{M_k}\,, 
&&\hat{X}_{{M_k}+1} = \acute{\phi}(\hat{X}_{M_k}, \theta_{{M_k}+1},U_{{M_k}+1})\,, \,\dots\,, \,
\hat{X}_0= \acute{\phi}(\hat{X}_1, \theta_0,U_0)\,.
\end{align*}
The (unknown) coupled $(\acute{X}_n)_{n\in\Z_-}$ in equilibrium satisfies 
\[
\check{X}_n \preccurlyeq \acute{X}_n \preccurlyeq \hat{X}_n\,, \quad n = M_k,\dots, 0\,.
\]
Since the update mechanism preserves the preorder $\preccurlyeq$ 
and that this constitutes an order on non-decreasing $c$-tuples, 
if $\check{X}_n = \hat{X}_n$ for some $n$ in $\{M_k,\dots, 0\}$
then  $\check{X}_m = \hat{X}_m$ for all $n\le m \le 0$, which
constitutes the event of coalescence, and in particular  
\[
\check{X}_0 = \acute{X}_0 = \hat{X}_0
\]
constitutes a perfect draw from the invariant law $\acute{\pi}$ of $\acute{X}$.

This back-off procedure must be repeated for $k\ge1$ until coalescence is achieved.
This happens at least as soon as the time to which the backward simulation is continued is
lesser than or equal to $N$ defined in \eqref{eq-empty-back}, 
which ensures some exponential moments for the duration,
see Theorem~\ref{DCFTP-empty}.

Note that in the above $U_{M_{k-1}+1}, \dots\,, U_0$ do not affect the $(\acute{Y}_n, \theta_{n+1})$, and can be redrawn anew since their reuse is not required for the coupling to be preserved.

This results in the following perfect simulation algorithm, written in mathematical fashion and not in pseudocode.
An actual implementation should follow Section~\ref{sec-sim-sym}.
The algorithm requires to choose a suitable back-off strategy \eqref{backoff}.

\begin{algorithm}[H]
\caption{DomCFTP with back-off and sandwiching}\label{alg-dcftp-sandw}
\begin{algorithmic}[1]
\renewcommand\algorithmicdo{}
\algnotext{EndFor}
\State draw  $Y_0$ from $\gamma^{\otimes c}$ and sort into $\acute{Y}_0$ 
\Comment see Section~\ref{rem-geom}
\For{$k=1$ to $K = \inf\{k\ge1 : \check{X}^{(k)}_0 = \hat{X}^{(k)}_0\}$}
	\State determine $M_k$
	\For{$n= M_{k-1}-1$ down to $M_k$}\label{a3-theta}
		\State draw $\acute{Y}_n$ using $\acute{Y}_{n+1}$
		\Comment see Theorem~\ref{theta-from-y} and \eqref{back-y}
		\State draw $\theta_{n+1}$ using $\acute{Y}_{n+1}$ and $\acute{Y}_n$
		\Comment see Theorem~\ref{theta-from-y}
	\EndFor
	\State set $\check{X}^{(k)}_{M_k} = 0$ and $\hat{X}^{(k)}_{M_k} = \acute{Y}_{M_k}$
	\For{$n=M_k+1$ up to $n=0$}\label{a3-up-low}
		\State draw $U_n$ uniformly in $[0,1]$ 
		\Comment reuse or redraw $U_n$ for $n\ge M_{k-1}+1$
		\State set $\check{X}^{(k)}_n =  \acute{\phi}(\check{X}^{(k)}_{n-1},\theta_n,U_n)$
		\Comment the $\theta_n$ obtained in Loop~\ref{a3-theta}
		\State set $\hat{X}^{(k)}_n =  \acute{\phi}(\hat{X}^{(k)}_{n-1},\theta_n,U_n)$ 
	\EndFor
\EndFor
\State\Return $\hat{X}^{(K)}_0$
\end{algorithmic}
\end{algorithm}

\begin{theorem}\label{DCFTP-up-low}
Let $\lambda < c$, and $\pi$ be the invariant law of the queue length process $X$ of a load balancing network
of Definition~\ref{def-class}. 
Algorithm~\ref{alg-dcftp-sandw} uses simulatable draws, terminates after a duration which has 
some exponential moments, and outputs a perfect draw from~$\pi$. 
Specifically, there exists $s>1$ such that $\E(s^K)<\infty$, thus $\E(K)<\infty$,
and $\hat{X}^{(K)}_0$ has law $\acute{\pi}$.
\end{theorem}

\begin{proof}
The proof follows from the preceding discussion.
\end{proof}

\begin{remark}
This algorithm can be used for the task described in Remark~\ref{rem-sim-1-paral}. The backward simulation 
of $(\acute{Y}_0,\dots,\acute{Y}_{M_k})$ and of $(\theta_0,\dots,\theta_{M_k+1})$ should be continued back in the past 
for $k\ge1$ until coalescence has been achieved for all the queue length processes of the networks under consideration. 
In the case of parallel computing, the backward simulation of $(\acute{Y}_n)_{n\in\Z_-}$ may be continued 
while forward lower and upper bound simulation at a $M_k$ are undertaken.
\end{remark}

\subsection{Assessment of the DomCFTP algorithms}\label{sec-Con-Ken-assess}

In~\cite[Sect.~6]{Connor-Kendall2015}, the authors assess and compare the two algorithms that they have 
developed for $M/G/c$. 

They do so mainly in the Markovian case $M/M/c$. 
The invariant law is well-known -- its normalizing constant is obtained from the well-tabulated Erlang $C$ formula -- and used for comparisons with the simulated draws.
The dominating Markov chain is given from $c$ i.i.d.\ $M/M/1$ queues as in the present paper,
and can be readily perfectly simulated without using the Pollaczek-Khintchine formula.

Simulations indicate that the more sophisticated back-off sandwiching algorithm takes considerably less 
time than the simpler emptying algorithm. 
The authors indicate that the discussion is far from complete and recall 
that the computational demands of the sandwiching algorithm are greater for similar durations. Some theoretic and heuristic arguments on the durations are also provided.

\section{Conclusion}

The main scope of this paper is to provide efficient DomCFTP methods in the truly multi-dimensional 
setting of a wide class of Markovian load balancing queueing networks. 
The policies we consider are quite general and natural, and include 
classic polices for parallel servers as well as policies implementing instantaneous
work stealing at service completions. These policies may respect topological constraints such as a directed graph
structure, which breaks the exchangeability assumption under which most asymptotic performance evaluation 
is achieved.

Such perfect simulation methods are valuable 
for Monte Carlo estimation of quantities of interest in equilibrium,  most notably for performance evaluation.
This is achieved on this infinite state space using an adequate coupling
providing a dominating Markov chain in an original preorder. It corresponds to an order up to permutations 
of the coordinates which is strictly weaker than the classic product order. 

The paper by Connor and Kendall~\cite{Connor-Kendall2015} has been an inspiration, and we are glad to 
have provided an answer to  Question~3 in its Section~7.
There are many remaining questions left unanswered, and we list a few.
\begin{enumerate}
\item
Provide mathematical evaluation of the benefits of the back-off and sandwiching algorithm over the emptying algorithm 
(\cite[Sect.~7 Question~2]{Connor-Kendall2015}).
\item
Implement niceties about DomCFTP such as recycling, interruptible algorithms, read-once, 
see~\cite[Chap.~5]{Huber2016},
simultaneous implementation for a suitable range of values of $c$ called ``omnithermal DomCFTP''
in \cite[Sect.~7 Question~4]{Connor-Kendall2015}), etc. 
We have noted that we may reuse or redraw the $U_n$ at will.
\item
Extend results to non-Markovian queueing networks with general service distributions or
general interarrival distributions or both, see~\cite[Sect.~9]{vanderBoor2022}, \cite{Blanchet2019}.
\item

Tell interesting things about asymmetric queueing networks, see~\cite[Sect.~9]{vanderBoor2022}, 
beyond the fact that the first two perfect simulation methods in this paper apply to the networks of Definition~\ref{def-class}.
\item
Find other wide natural classes of networks, couplings, and orders for which the ideas in this work
can be adapted.
\end{enumerate}

\backmatter

\bmhead{Acknowledgments}
The author thanks Stephen Connor for a private communication about the mean run time of 
the DomCFTP algorithms in~\cite{Connor-Kendall2015}. This paper was an inspiration for us.


\end{document}